\documentclass{article}
\usepackage[margin=1in]{geometry}
\usepackage{authblk}
\RequirePackage{amsthm,amsmath,amsfonts,amssymb}

\RequirePackage[numbers]{natbib}
\RequirePackage[colorlinks,citecolor=blue,urlcolor=blue]{hyperref}
\RequirePackage{graphicx}
\usepackage{subcaption}
\usepackage{tikz}
\usepackage{comment}
\numberwithin{equation}{section}
\numberwithin{figure}{section}

\newtheorem{thm}{Theorem}[section]
\newtheorem{lem}[thm]{Lemma}

\newtheorem{cor}[thm]{Corollary}

\newcommand{\E}[1]{{\mathbf E}\left[#1\right]}
\newcommand{\p}[1]{{\mathbf P}\left(#1\right)}
\newcommand{\I}[1]{{\mathbf 1}_{[#1]}}

\newcommand{\Bin}[1]{\mathrm{Bin}(#1)}
\newcommand{\Poi}[1]{\mathrm{Poi}(#1)}

\newcommand{\eqdist}{\stackrel{d}{=}}

\title{Limit theorems for Randić index for Erd\H{o}s-Rényi graphs}

\author[1]{Laura Eslava} 
\author[2]{Saylé Sigarreta}
\author[1]{Arno Siri-Jégousse}

\affil[1]{UNAM, IIMAS, Departamento de Probabilidad y Estadística,
México.}
\affil[2]{BUAP, Departamento de Probabilidad y Estadística,
México.}

\begin{document}

\maketitle

\begin{abstract}
We prove that the generalized Randić index over graphs following the Erd\H{o}s-Rényi model, for both the sparse and dense regimes, is concentrated around its mean when the number of vertices tends to infinity.
\end{abstract}


\textbf{MSC classes}: 05C50, 05C80

\section{Introduction}

A molecular graph, within the domain of chemistry, commonly represents the structure of a molecule graphically, with atoms as nodes and chemical bonds as edges connecting these nodes. Comprehending these structural representations is pivotal in pharmaceuticals for designing drugs, given that the arrangement of atoms and bonds significantly influences the molecule's properties \cite{WCRSNYH,MN}. 
In this sense, a numerical value, irrespective of the particular node and edge labeling, which encapsulates specific characteristics of the molecular graph, is termed a graph invariant or topological index.  Hence, various topological indices and their mathematical properties have been explored \cite{G,GF}. In particular,  the Randić index emerges as the foremost extensively studied, commonly employed, and broadly acknowledged topological index \cite{LG,LS}; this is due to its close correlation with numerous chemical properties, including boiling points, surface area, solubility in water of alkanes \cite{KH,SD,GGS}. It was initially termed as the ``branching index" when introduced by chemist Milan Randić in 1975 \cite{R} as
\begin{equation*}
    R(G)=\sum_{u\sim v } \frac{1}{\sqrt{d_u d_v}};
\end{equation*}
where the sum is over all the pairs of connected vertices $u\sim v$ of a graph $G$ and $d_u,d_v$ stand for their respective degree (i.e. their number of connected neighbours in the graph).  
This index was later generalized by Bollobás and Erd\H{o}s in 1998 \cite{BE}, for $\alpha\in \mathbb{R}$,
\begin{equation}\label{e2}
    R_{\alpha}(G)=\sum_{u\sim v } (d_u d_v)^{\alpha}.
\end{equation}

On the other hand, graph theory has shown to have applicability in representing and modeling various real-world systems as random graphs, as many of their crucial features are challenging to capture using deterministic models. In particular,  a model of random graphs with a fixed number of vertices $n$ is a stochastic mechanism for determining which of the  $n(n - 1)/2$ potential edges actually arise. In 1959, the Erd\H{o}s-Rényi (ER) random graph \cite{ER}
was defined. In this graph, every pair of vertices is independently linked with  probability  $p$. When $p$ remains constant, it is termed as a dense ER graph. Conversely, if $p = O(1/n)$, it is referred to as a sparse ER graph.  The properties of ER graphs have been thoroughly examined, including their largest independent sets, chromatic number, Hamiltonian paths, connected components, automorphisms, among others \cite{JLR}. Additionally,  in terms of applications, ER is essentially akin to an elementary model of an epidemic known as the Reed-Frost model, the findings regarding ER are beneficial for understanding evolutionarily stable strategies and, recently, some slight modifications of the original model serve as a significant foundation to model interactions within communities \cite{PC,SKP}. In simpler terms, communities represent tightly-knit groups of nodes where interactions are frequent among members but less so with nodes outside the community. Hence, community analysis can reveal important patterns, decomposing large collections of interactions into more meaningful components.  In general, although ER graphs are commonly recognized as suboptimal models for interaction networks, their simplicity, mathematical manageability, and relations with alternative models possessing real-world properties render them a valuable asset in network research and analysis.

Under the same line, several topological indices have been investigated within the ER model, employing computational or analytical approaches \cite{LLS, CLH, MMRS, DHHIR}; since, 
for instance, comprehensive knowledge of various indices pertaining to the model could function as indicators of whether a network conforms to the Erd\H{o}s-Rényi random graph model. Similarly, the generalized Randić index has been scrutinized across various random models \cite{FMP, ZW, SSC,SSC2}. Moreover, as points of intersection, the generalized Randić index has been examined within the ER model \cite{Y, MRS,DMRSV}.  Specifically, the generalized Randić index has been employed to estimate graph robustness and, at the same time, ER graphs have served as valuable benchmarks for evaluating graph robustness \cite{DMRSV}. This measure, which indicates a graph's capacity to sustain connectivity amidst node and edge loss, aids in examining the resilience of various systems.
Inspired by the foregoing, the objective of this paper is to examine the asymptotic properties of the generalized Randić index in both dense and sparse scenarios through the application of the second  moment method. Additionally, it aims to corroborate the empirical findings of previous studies and to show consistency with the theoretical results derived through different approaches.

\section{Large number approximations}

In this study, we exclusively focus on simple finite undirected graphs. A simple undirected graph is represented as $G=(V,E)$, where $E\subseteq V\times V$ is a set ensuring that $(v,v)\notin E$ for all $v\in V$, and if $(v,u)\in E$, then $(u,v)\in E$. Here, $V$ signifies the vertex set and $E$ represents the edge set. In a graph $G$, two vertices $u$ and $v$ are considered adjacent, or neighbors, if $(u, v)$ forms an edge in $G$, denoted as $u \sim v$. If all the vertices of $G$ are pairwise adjacent, then $G$ is complete. A complete graph on $n$ vertices is denotaded as $K_n$. The degree (or valency) $d_v$ of a vertex $v$ is the number of edges at $v$; this is equal to the number of neighbours of $v$. The number $\delta(G):=\min \{d_v \mid v \in V\}$ is the minimum degree of $G$, the number $\Delta(G):=\max \{d_v \mid v \in V\}$ its maximum degree. 

We will work with the Erd\H{o}s-Rényi random graph $G_{n,p}$ with parameters $n\in \mathbb{N}$ and $p\in (0,1)$. Vertices are labelled by $[n]=\{1,2,\ldots, n\}$. 
In this work we consider the generalized Randić index  $R_\alpha(G_{n,p})$ with  $\alpha\in \mathbb{R}$, as defined in \eqref{e2}. Recall that for the complete graph, $R_{\alpha}(K_n)=\frac{n(n-1)^{1+2\alpha}}{2}$.

\subsection{Sparse graphs}

Let us first consider the case where $p=\lambda/n$ for $\lambda\in (0,\infty)$.
Define, for a random variable $X=X_\lambda \eqdist \Poi{\lambda}$ and $\alpha\in  \mathbb{R}$,  
\begin{align*}
    c[\alpha,\lambda]=c_{\alpha,\lambda}=\E{(1+X)^{\alpha}}.
\end{align*}

\begin{thm}\label{thm:sparse}
    For  $\lambda\in (0,\infty)$, the generalized Randić index of $G_{n,p}$ with parameter $\alpha\in \mathbb{R}$ is asymptotically concentrated around its mean. More precisely, as $n\to \infty$,     \begin{align*}
        \frac{2R_\alpha(G_{n,p})}{n}=\frac{R_\alpha(G_{n,p})}{R_{-1/2}(K_{n})}\to \lambda  c_{\alpha,\lambda}^2,
    \end{align*}
    in probability. Moreover, we observe that $\lambda c^2_{\alpha,\lambda}$ exhibits the same asymptotic behavior as $\lambda^{1+2\alpha}$ as $\lambda$ approaches infinity.
\end{thm}
The proof of Theorem~\ref{thm:sparse} is based on the second moment method and the following asymptotic expressions. 
\begin{lem}\label{lem:degrees}
Let $\lambda\in (0,\infty)$, $\alpha\in \mathbb{R}$ and consider $G_{n,p}$ with $p=\lambda/n$. As $n\to \infty$,    
\begin{align*}
    \E{(d_1d_2)^{\alpha}|1\sim 2}
    \to  c_{\alpha,\lambda}^2 \qquad \text{and} \qquad
    \E{(d_1d_2d_3d_4)^{\alpha}|1\sim 2,3\sim 4}
    \to  c_{\alpha,\lambda}^4.
\end{align*}
\end{lem}

\begin{proof}
 Let $X^{(n)}_i$, $i=1,2$ be independent with $\Bin{n-2,\lambda/n}$ distribution. Conditional on $\{1\sim 2\}$, $d_1\eqdist 1+X^{(n)}_1$ and $d_2\eqdist 1+X^{(n)}_2$, thus we have
\begin{align*}
    \E{(d_1d_2)^{\alpha}|1\sim 2}
    =\E{\left(1+X^{(n)}_1\right)^{\alpha}}\E{\left(1+X^{(n)}_2\right)^{\alpha}}=\E{\left(1+X^{(n)}_1\right)^{\alpha}}^2.
\end{align*}
Since $X_1^{(n)}$ converges in distribution to $X$, a $\Poi{\lambda}$ r.v. and $f(x)=(1+x)^{\alpha}$ is  continuous on $[0,\infty)$, then,  $(1+X_1^{(n)})^{\alpha}$ converges in distribution to $(1+X)^{\alpha}$. At this point, let us check that $\{(1+X_1^{(n)})^{\alpha}\}_{n}$ is uniformly integrable. To see this, we have that 

$$\E{\left(1+X_1^{(n)}\right)^{\alpha}}=\displaystyle\sum_{k=0}^{n-2} (1+k)^{\alpha} \binom{n-2}{k} \left(\frac{\lambda}{n}\right)^k \left(1- \frac{\lambda}{n}\right)^{n-2-k}\leq\displaystyle\sum_{k=0}^{n-2} (1+k)^{\alpha} \frac{\lambda^k}{k!}. $$
Considering the convergence of the series in the r.h.s. by d'Alembert's criterion, we deduce that 

$$\sup_{n\in \mathbb{N}} \E{\left(1+X_1^{(n)}\right)^{\alpha}} <\infty.$$
This observation leads to the conclusion that $\{(1+X_1^{(n)})^{\alpha}\}_{n}$ is uniformly integrable, therefore, as $n\to \infty$
\begin{align}\label{k}
    \E{(d_1d_2)^{\alpha}|1\sim 2}
    \to \E{(1+X)^{\alpha}}^2= c_{\alpha,\lambda}^2.
\end{align}
For the second statement, let $Y_i^{(n)}$, $1\le i\le 4$ be independent with $\Bin{n-4,\lambda/n}$ distribution. Similarly, conditional on  $\{1\sim 2, 3\sim 4\}$, $d_i\eqdist 1+\delta_i+Y_i^{(n)}$ where $\{\delta_i\}_{i\in [4]}$ are correlated random variables: $\delta_{i}=\I{i\sim 3}+\I{i\sim 4}$ and $\delta_{i}=\I{i\sim 1}+\I{i\sim 2}$, for $i=1,2$ and $i=3,4$, respectively. 
So, if let  $A$ denote the event that there is at least one edge connecting the vertices in $\{1,2\}$ and $\{3,4\}$; since $A^c$ tells us that $\delta_i=0$ for all $1\le i\le 4$, it follows that

\begin{align*}
& \E{(d_1d_2d_3d_4)^{\alpha}|1\sim 2,3\sim 4} \\
&= \E{(d_1d_2d_3d_4)^{\alpha}|1\sim 2,3\sim 4,A^c}\p{A^c}+\E{(d_1d_2d_3d_4)^{\alpha}|1\sim 2,3\sim 4,A}\p{A}, \\
&= \E{\left(1+Y_1^{(n)}\right)^{\alpha}}^4(1-\frac{\lambda}{n})^4+\E{(d_1d_2d_3d_4)^{\alpha}|1\sim 2,3\sim 4,A}(1-(1-\frac{\lambda}{n})^4),
\end{align*}
    which completes the proof since the first summand converges using a similar argument as in (\ref{k}), while the second summand goes to 0.
\end{proof}

\begin{proof}[Proof of Theorem \ref{thm:sparse}]
    In order to use the second moment method, it suffices to verify that, as $n\to \infty$,
\begin{align*}
  \frac{2\E{R_\alpha(G_{n,p})}}{n\lambda  c_{\alpha,\lambda}^2}\to 1 \qquad \text{and}\qquad \frac{\E{R^2_\alpha(G_{n,p})}}{\E{R_\alpha(G_{n,p})}^2}\to 1.  
\end{align*}
The first limit follows from the exchangeability of the vertices in $G_{n,p}$ and Lemma~\ref{lem:degrees}, since  
\begin{align*}
    \E{R_\alpha(G_{n,p})}&=\binom{n}{2} \E{(d_1d_2)^{\alpha}\I{1\sim 2} }
    =\frac{\lambda(n-1)}{2} \E{(d_1d_2)^{\alpha}|1\sim 2}.
\end{align*}
For the second limit, consider  $\E{R^2_\alpha(G_{n,p})}$, which equals
\begin{align*}
6\binom{n}{4} \E{(d_1d_2d_3d_4)^{\alpha} \I{1\sim 2}\I{3\sim 4}}+6 \binom{n}{3} \E{(d_1^2d_2d_3)^{\alpha} \I{1\sim 2}\I{1\sim 3}}+\binom{n}{2} \E{(d_1d_2)^{2\alpha} \I{1\sim 2}}.
\end{align*}
While 
\begin{align*}
6 \binom{n}{3} \E{(d_1^2d_2d_3)^{\alpha} \I{1\sim 2}\I{1\sim 3}}+\binom{n}{2} \E{(d_1d_2)^{2\alpha} \I{1\sim 2}}=O(n),
\end{align*}
the leading term, by Lemma~\ref{lem:degrees}, satisfies
\begin{align*}
6\binom{n}{4} n^{-2} \E{(d_1d_2d_3d_4)^{\alpha} \I{1\sim 2}\I{3\sim 4}}= (1+O(n^{-1})) \frac{\lambda^2}{4} \E{(d_1d_2d_3d_4)^{\alpha} |\I{1\sim 2}\I{3\sim 4}} \to \frac{\lambda^2 c_{\alpha,\lambda}^2}{4};
\end{align*}
which completes the proof of the second limit. 

For completeness: Since $\frac{\E{R_\alpha(G_{n,p})}}{n}\to \frac{\lambda  c_{\alpha,\lambda}^2}{2}$ as $n\to \infty$, it suffices to prove using Chebyshev's inequality, that for $\varepsilon>0$,
\begin{align*}
    \p{|R_\alpha(G_{n,p})-\E{R_\alpha(G_{n,p})}|\ge \varepsilon\E{R_\alpha(G_{n,p})}}
    \le \varepsilon^{-2} \frac{\mathrm{Var}(R_\alpha(G_{n,p}))}{\E{R_\alpha(G_{n,p})}^2}= \varepsilon^{-2}\left(\frac{\E{R^2_\alpha(G_{n,p})}}{\E{R_\alpha(G_{n,p})}^2}-1 
    \right)
    \to 0;
\end{align*}
as $n\to \infty$.

To confirm the last assertion for $\alpha \geq 0$, observe that $\lambda\E{(1+X)^{\alpha}}^2 \geq \lambda$. For $\alpha <0$, coincidentally, using the same approach that will be detailed in Lemma \ref{lem:densedegree}, but employing the Chernoff Poisson bounds, we can verify that when $\lambda \to \infty$, $\lambda^{-\alpha} c_{\alpha,\lambda} \to 1$ (see, e.g., Section 2.2 in \cite{BLM}). Consequently, the last statement is true.
\end{proof}

\subsection{Dense graphs}
The following analogous theorem holds for the case $p\in(0,1)$ fixed. 

\begin{thm}\label{thm:dense}
    For $p\in(0,1)$, the generalized Randić index of $G_{n,p}$ with parameter $\alpha\in \mathbb{R}$ is asymptotically concentrated around its mean. More precisely, as $n\to \infty$,
    \begin{align*}
        \frac{2R_\alpha(G_{n,p})}{n^{2(1+\alpha)}}\sim\frac{R_\alpha(G_{n,p})}{R_\alpha(K_{n})}\to {p^{1+2\alpha}},
    \end{align*}
    in probability.
\end{thm}

The above theorem relies on the following lemma.
\begin{lem}\label{lem:densedegree}
Let $\alpha\in \mathbb{R}$, $p\in (0,1)$ and consider $G_{n,p}$. As $n\to \infty$,    
\begin{align*}
    n^{-2\alpha}\E{(d_1d_2)^{\alpha}|1\sim 2}
    \to p^{2\alpha} \qquad \text{and} \qquad
    n^{-4\alpha}\E{(d_1d_2d_3d_4)^{\alpha}|1\sim 2,3\sim 4}
    \to p^{4\alpha}.
\end{align*}
\end{lem}

\begin{proof}
Let $X^{(n)}\eqdist \Bin{n-2,p}$ and $Y^{(n)}\eqdist \Bin{n-4,p}$, it follows that 
  
  \begin{align*}
    \E{(d_1d_2)^{\alpha}|1\sim 2}
    =\E{\left(1+X^{(n)}\right)^{\alpha}}^2,
\end{align*}
and

  \begin{align*}
   \E{\left(c_1+Y^{(n)}\right)^{\alpha}}^4  \leq\E{(d_1d_2d_3d_4)^{\alpha}|1\sim 2,3\sim 4}
    \leq\E{\left(c_2+Y^{(n)}\right)^{\alpha}}^4,
\end{align*}
where $c_1=1,3$ and $c_2=3,1$, for $\alpha \geq  0$  and $\alpha < 0$, respectively. Thus, the proof boils down to establishing, for $X^{(n)}$ (and similarly, for $Y^{(n)}$) that 
    \begin{align}\label{eq:Xdense}
    n^{-\alpha} \E{\left(1+X^{(n)}\right)^{\alpha}}\to p^{\alpha}. 
    \end{align}
    According to the Weak Law of Large Numbers (WLLN), we have that $X^{(n)}/n$ converges to $p$ in probability. Hence, $(1+X^{(n)})/n$ converges to $p$ in probability and since  $f(x)=x^{\alpha}$ is continuous on $(0,\infty)$, it is verified that $n^{-\alpha}(1+X^{(n)})^{\alpha}$ converges to $p^{\alpha}$ in probability. At this point, similarly, we are going to verify that $\{((1+X_1^{(n)})/n)^{\alpha}\}_{n }$ is uniformly integrable. It is worth noting that for $\alpha \geq 0$, verifying the statement is straightforward as $(1+X^{(n)})/n \leq 1$, for all $n$. Let's check the case  $\alpha < 0$: To see this, let $\varepsilon \in (0,p)$ and define $A^+_{\varepsilon,n}=\{ X^{(n)}\ge (p+\varepsilon)(n-2)\}$ and $A^-_{\varepsilon,n}=\{ X^{(n)}\le (p-\varepsilon)(n-2)\}$. Using Chernoff bounds (see, e.g. \cite[Theorem 2.1]{JLR} and \cite[Corollary 2.2]{JLR}) and that $\alpha < 0$, 
\begin{align*}
    n^{-\alpha} \E{(1+X^{(n)})^{\alpha}\I{A^+_{\varepsilon,n}}}
    &\le n^{-\alpha}\p{A^+_{\varepsilon,n}} \le 
    n^{-\alpha}e^{-\frac{\varepsilon^2(n-2)}{3p}},\\
    n^{-\alpha} \E{(1+X^{(n)})^{\alpha}\I{A^-_{\varepsilon,n}}}
    &\le n^{-\alpha} \p{A^-_{\varepsilon,n}} \le n^{-\alpha}e^{-\frac{\varepsilon^2(n-2)}{2p}}.
\end{align*}
Thus, 

$$\displaystyle\lim_{n\to \infty}  n^{-\alpha} \E{(1+X^{(n)})^{\alpha}\I{A^+_{\varepsilon,n}}}=\displaystyle\lim_{n\to \infty}  n^{-\alpha} \E{(1+X^{(n)})^{\alpha}\I{A^-_{\varepsilon,n}}}=0.$$
On the other hand, 
$$n^{-\alpha} \E{(1+X^{(n)})^{\alpha}\I{(A^+_{\varepsilon,n}\cup A^-_{\varepsilon,n})^c}} \leq \left(\frac{1+(p-\varepsilon)(n-2)}{n}\right)^{\alpha}.$$
Therefore,
\begin{align*}\label{eq:neg}
   n^{-\alpha} \E{(1+X^{(n)})^{\alpha}}\leq (p-\epsilon)^{\alpha}+2\epsilon,
\end{align*}
for $n \geq N_{\epsilon}$. This tells us that the desired sequence is uniformly integrable, since $$\sup_{n\in \mathbb{N}} \E{\left(\frac{1+X^{(n)}}{n}\right)^{\alpha}} <\infty,$$ remember that $\varepsilon\in (0,p)$ is arbitrary.  Therefore, the proof is completed.
\end{proof}

\begin{proof}[Proof of Theorem~\ref{thm:dense}]
Analogously to the proof of Theorem~\ref{thm:sparse}, the exchangeability of the vertices in $G_{n,p}$ gives
\begin{align*}
    \E{R_\alpha(G_{n,p})}&=\binom{n}{2} p \E{(d_1d_2)^{\alpha}|1\sim 2},
\end{align*}
and $ \E{R^2_\alpha(G_{n,p})}$ equals to

   $$6\binom{n}{4} p^2\E{(d_1d_2d_3d_4)^{\alpha} |1\sim 2,3\sim 4}+ 6 \binom{n}{3}p^2 \E{(d_1^2d_2d_3)^{\alpha} |1\sim 2,1\sim 3}+\binom{n}{2}p \E{(d_1d_2)^{2\alpha}|1\sim 2}.$$
So that by Lemma~\ref{lem:densedegree}, as $n\to \infty$, 
\begin{align*}
    \frac{\E{R_\alpha(G_{n,p})}}{n^{2(1+\alpha)}} = (1+O(n^{-1}))\frac{p n^{-2\alpha}}{2}\E{(d_1d_2)^{\alpha}|1\sim 2} \to \frac{p^{1+2\alpha}}{2},
\end{align*}
and

\begin{align*}
    \frac{\E{R^2_\alpha(G_{n,p})}}{n^{4(1+\alpha)}} = (1+O(n^{-1}))\frac{p^2 n^{-4\alpha}}{4}\E{(d_1d_2d_3d_4)^{\alpha}|1\sim 2,3\sim 4} + O(n^{-1})\to \frac{p^{2+4\alpha}}{4},
\end{align*}
since, as in Lemma~\ref{lem:densedegree}

 \begin{align*}
    n^{-4\alpha} \E{(d_1^2d_2d_3)^{\alpha} |1\sim 2,1\sim 3} \to p^{4\alpha},
    \end{align*}
and
 \begin{align*}
    n^{-4\alpha}\E{(d_1d_2)^{2\alpha} |1\sim 2} \to p^{4\alpha},
    \end{align*}
Therefore, the proof is established using the second moment method, as in the proof of Theorem~\ref{thm:sparse}.
\end{proof}

\section{Discussion}

Limit theorems for the classic Randić connectivity index can be derived from Theorems \ref{thm:sparse} and \ref{thm:dense} by taking $\alpha=-1/2$. 
In the dense case, we observe that its asymptotic behavior is independent of the parameter $p$, i.e,
its determination comes only from the structure; whereas the limits of its variations do fluctuate: as a function of $p$, the long-term behavior of the generalized Randić index is increasing  for $-1/2 < \alpha $  and  decreasing for $\alpha < -1/2$. By the way, as a function of $\alpha$, the long-term behavior of the generalized Randić index is decreasing for $p \in (0,1)$. Throughout this section we define 
$L[\alpha,\lambda]= \lambda c^2_{\alpha,\lambda}.$ 

\begin{cor}\label{cd1}
    For both the dense and sparse case of ER graphs, the classic Randić index of $G_{n,p}$ is asymptotically concentrated around its mean. More precisely, as $n\to \infty$,

      \begin{align*}
        \frac{R_{-1/2}(G_{n,p})}{R_{-1/2}(K_n)}\to L,
    \end{align*}
    in probability, where 
    $L=1$ for the dense case $(p\in (0,1))$ and 
    $L=L[-1/2,\lambda]$ for sparse case $(p=\lambda/n, \lambda> 0)$. Moreover, $L[-1/2,\lambda]$ converges to 1 as $\lambda \to \infty$.
\end{cor}

We can compare Corollary \ref{cd1} with the results of \cite{MRS}, where the authors conducted a computational analysis of $R_{-1/2}(G)$ in the dense Erd\H{o}s-Rényi model. Specifically, they plotted the average Randić index against the probability $p$ for various orders  of $n$ and noted a transition from zero to $n/2$, with the rate of change increasing proportionally to $n$. This observation aligns well with the previous finding in the dense scenario. Additionally, when they re-plotted the curves as a function of $np$, they noticed that curves representing different graph sizes $n$ overlapped, forming a single universal curve. This observation, and even the asymptotic behaviour of the curve, can be further understood by considering our sparse results, which indicate that when the product $np=\lambda$ is constant, the classic Randić index over $n/2$ goes to a fixed value, and that it goes to 1 when $\lambda$ goes to infinity. 

In fact, plotting $L[-1/2,\lambda]$ against \( \lambda \) using a logarithmic scale (refer to Figure \ref{f1}), we observe a curve consistent with the universal curve depicted in Figure 2(b) of \cite{MRS}. This alignment enables the validation of the methodology proposed for predicting the classic Randić index's value. Specifically, it delineates three distinct regimes: the predominantly isolated vertices regime (\( \lambda < 0.01 \)), the transitional regime (\( 0.01 < \lambda < 10 \)), and the regime characterized by nearly complete graphs (\( \lambda > 10 \)).

\begin{figure}[h!]

   \centering
    \includegraphics[width=0.6\textwidth]{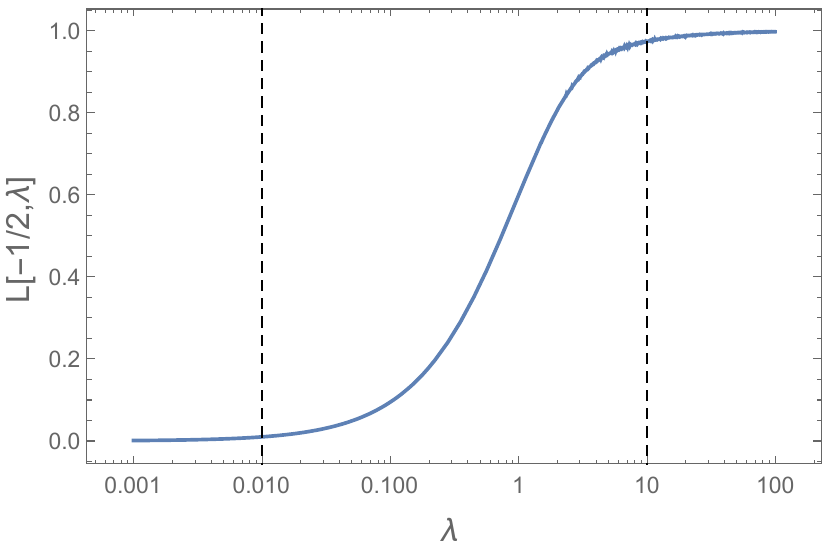}
    \caption{ $L[-1/2,\lambda]$  as a function of $\lambda$, where the expectation was computed by averaging over 2000 random variables $Poi(\lambda)$.} 
    \label{f1}
\end{figure}
When $\alpha=-1$, $R_{-1}$ is known as the second modified Zagreb index; which is related to the eigenvalues of the normalized Laplacian matrix of the graph \cite{C}. Thus, limit theorems concerning the second modified Zagreb index can be obtained from Theorems \ref{thm:sparse} and \ref{thm:dense} by setting $\alpha=-1$.
\begin{cor}\label{cd2}
    The modified second Zagreb index of $G_{n,p}$ is asymptotically concentrated around its mean. More precisely, as $n\to \infty$,

      \begin{align*}
        R_{-1}(G_{n,p})\to \frac{1}{2p},
    \end{align*}
    in probability, for the dense case.  Moreover, $L[-1,\lambda]$ converges to 0 as $\lambda \to \infty$.
\end{cor}
The following expression was derived in \cite{DMRSV}, for the expected value of $R_{-1}$ in the dense case,

$$\mathbb{E}[R_{-1}(G_{n,p})]=\frac{\left(1-(1-p)^{n+1}\right)^2}{2p}.$$
This result aligns with the previously mentioned observation that $R_{-1}$ asymptotically converges to its expectation.

Inspired by the above analysis and in order to complement the main results, $L[\alpha,\lambda]$ and $c[\alpha,\lambda]$  are plotted as functions of $\lambda$ for some values of $\alpha$ in Figure \ref{f2.1} and Figure \ref{f2.2}, respectively. Please note that the plot of $L[-1/2,\lambda]$ has already been depicted in Figure \ref{f1}. Here, the average is computed over a sample of 2000 random variables $Poi(\lambda)$.

\begin{figure}[h!]
  \centering
  \begin{minipage}[t]{0.3\textwidth}
    \centering
    \includegraphics[width=\linewidth]{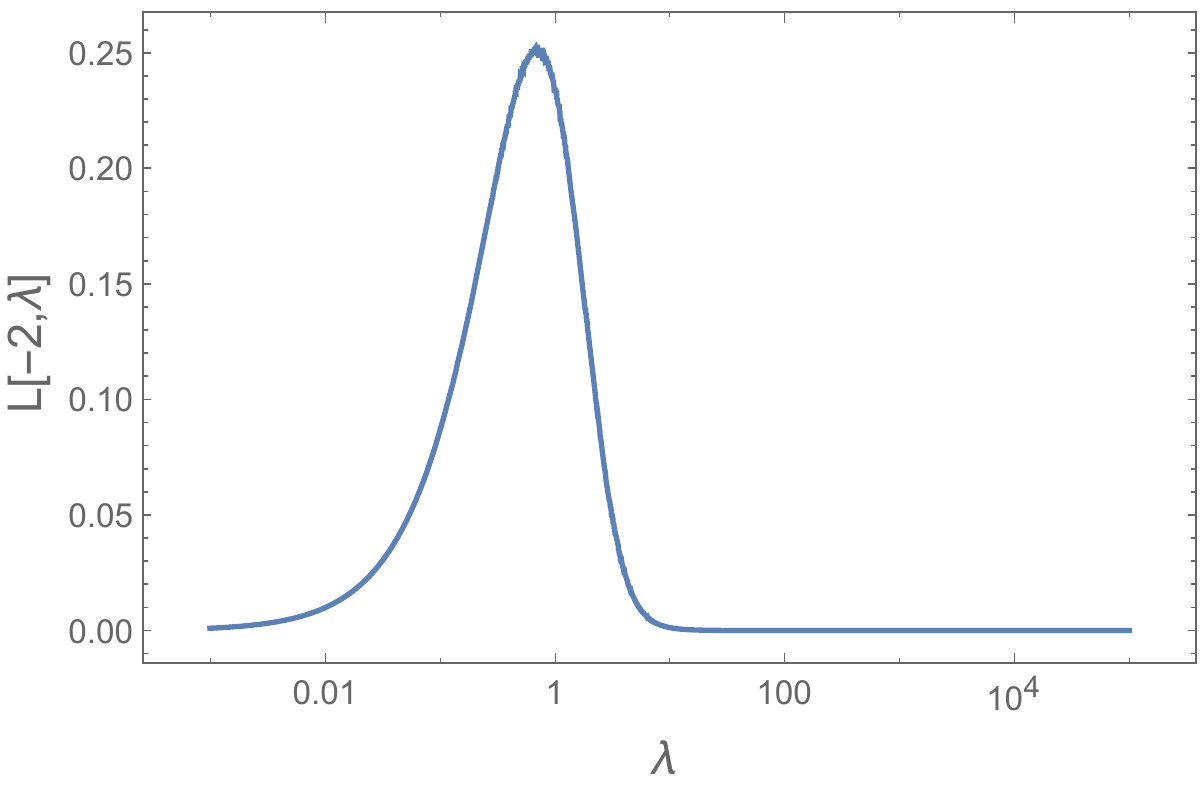}
    \caption*{$\alpha=-2$}
  \end{minipage}%
  \begin{minipage}[t]{0.3\textwidth}
    \centering
    \includegraphics[width=\linewidth]{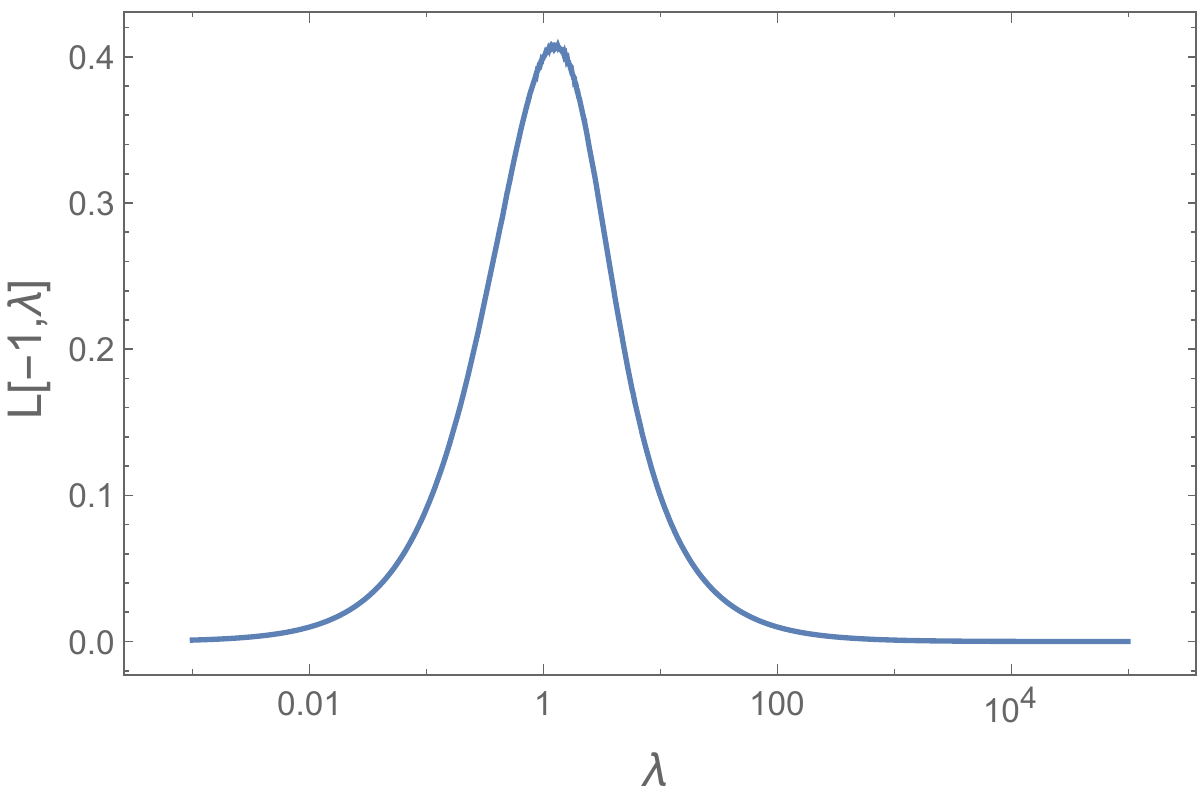}
    \caption*{$\alpha=-1$}
  \end{minipage}%
  \begin{minipage}[t]{0.3\textwidth}
    \centering
    \includegraphics[width=\linewidth]{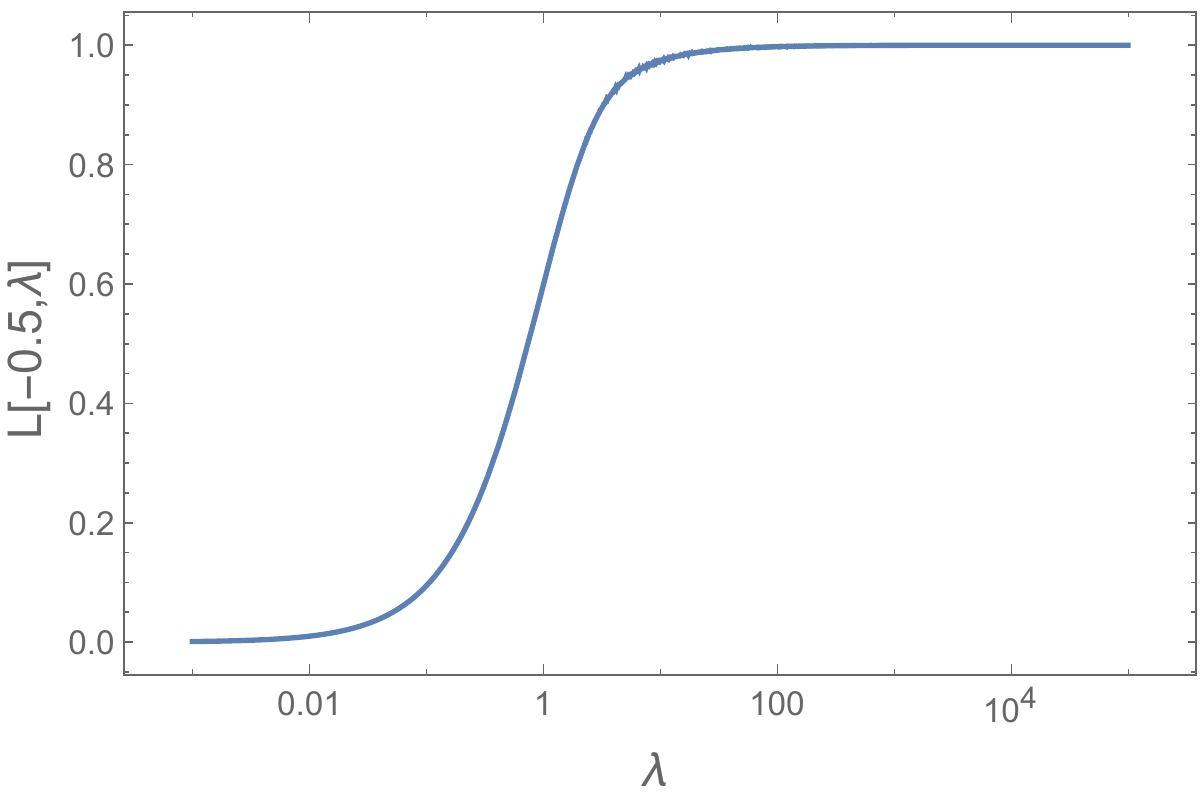}
    \caption*{$\alpha=-0.5$}
  \end{minipage}

  \vspace{0.5cm} 

  \begin{minipage}[t]{0.3\textwidth}
    \centering
    \includegraphics[width=\linewidth]{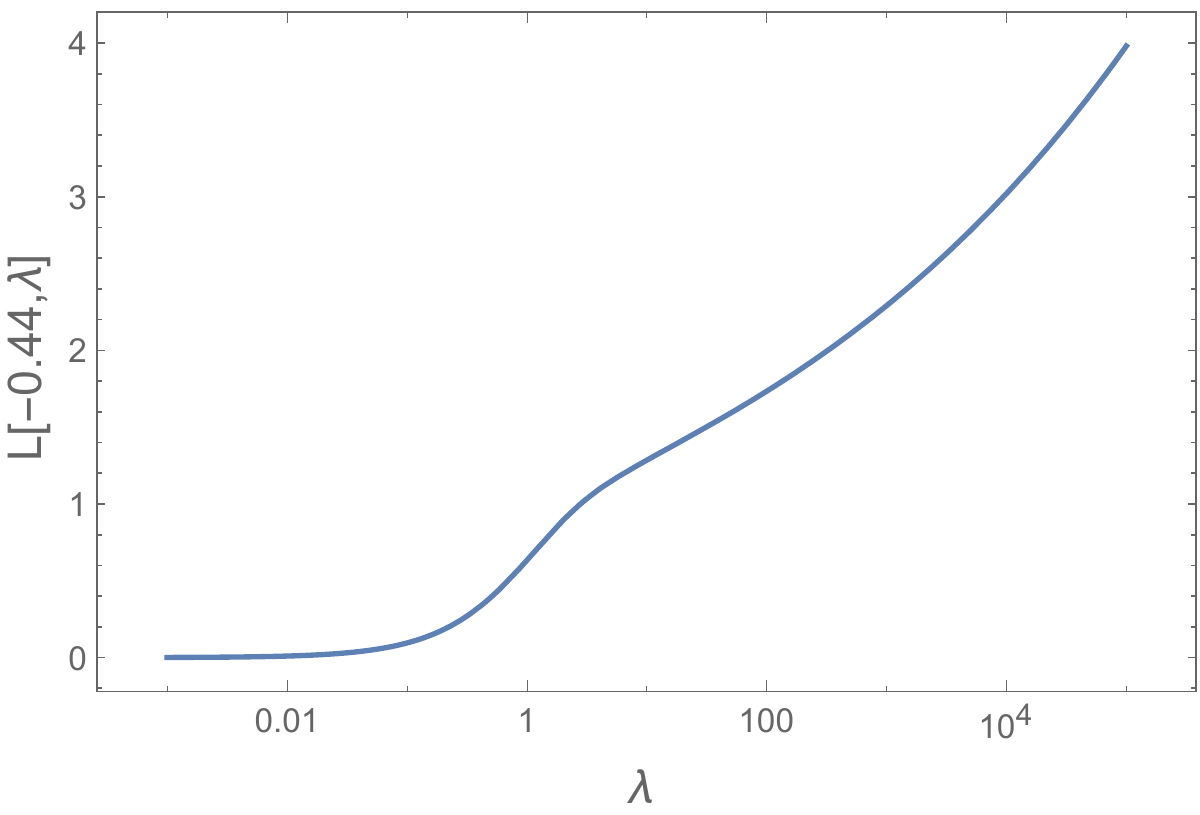}
    \caption*{$\alpha=-0.44$}
  \end{minipage}%
  \begin{minipage}[t]{0.3\textwidth}
    \centering
    \includegraphics[width=\linewidth]{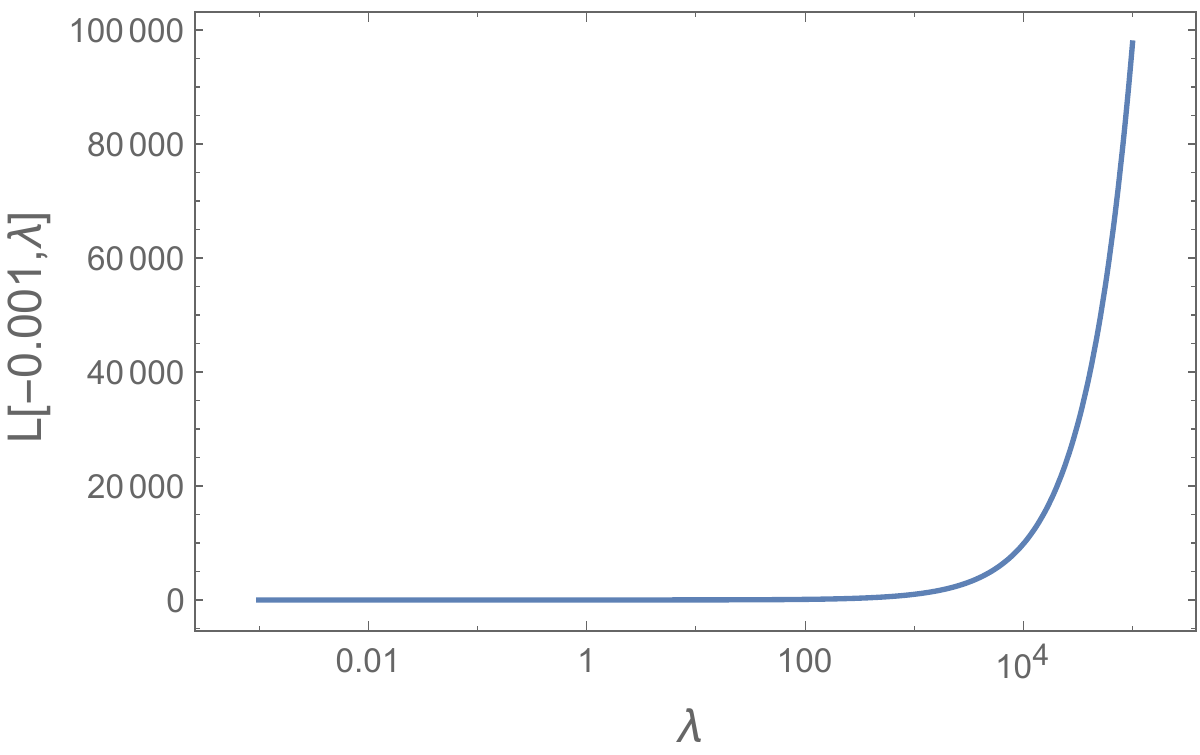}
    \caption*{$\alpha=-0.001$}
  \end{minipage}%
  \begin{minipage}[t]{0.3\textwidth}
    \centering
    \includegraphics[width=\linewidth]{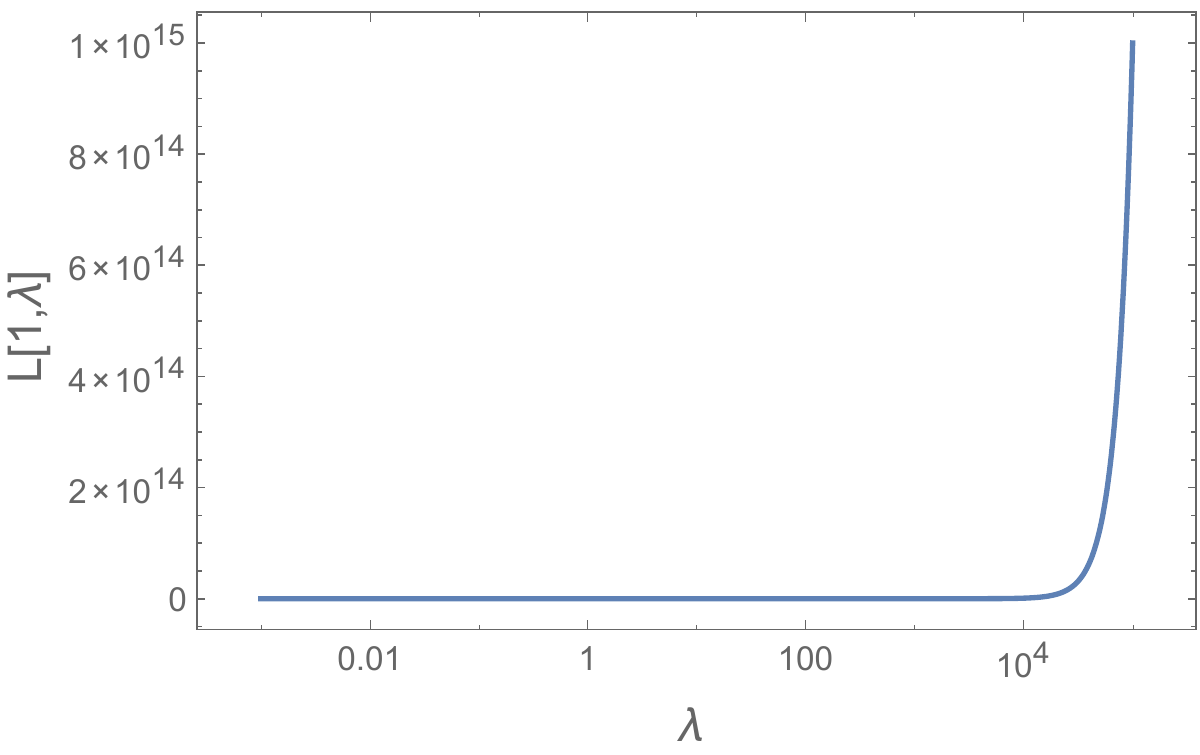}
    \caption*{$\alpha=1$}
  \end{minipage}%

  \caption{$L[\alpha,\lambda]$ as a function of $\lambda$, where the average was calculated by averaging over 2000 random variables $Poi(\lambda)$.}
  \label{f2.1}
\end{figure}

  \begin{figure}[h!]
  \centering
  \begin{minipage}[t]{0.3\textwidth}
    \centering
    \includegraphics[width=\linewidth]{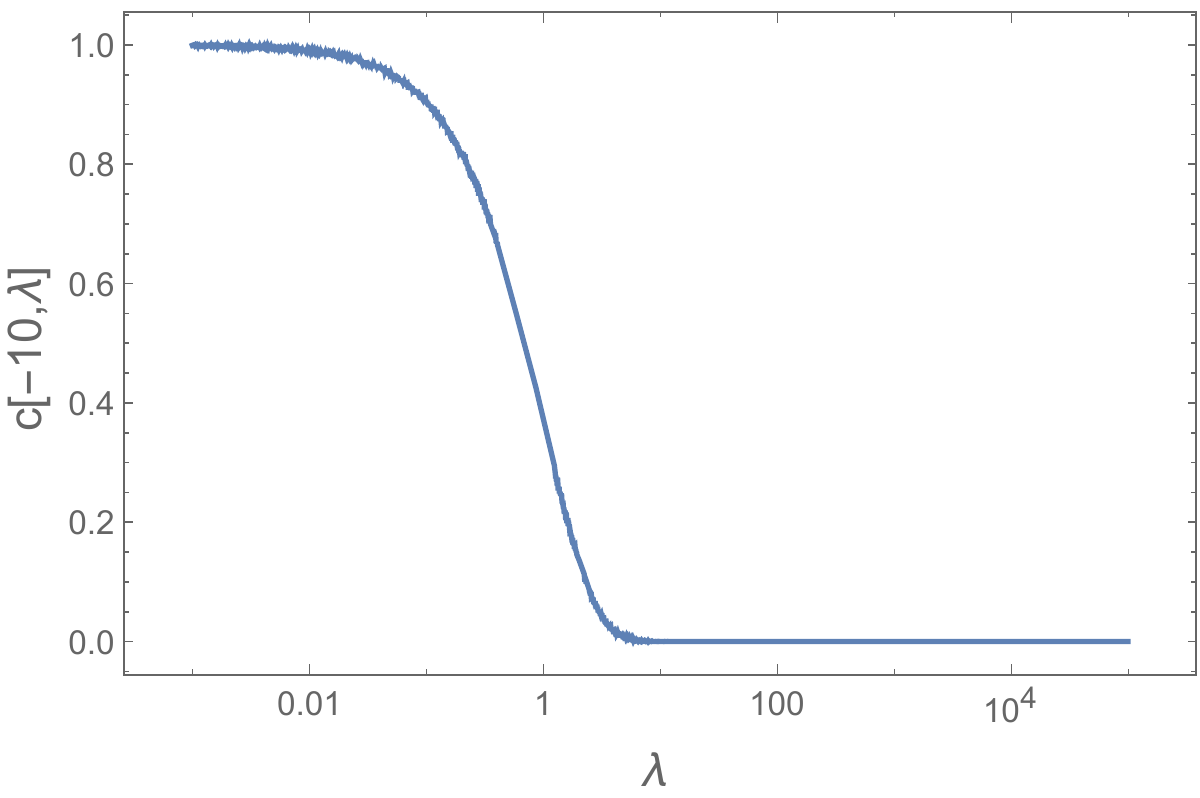}
    \caption*{$\alpha=-10$}
  \end{minipage}%
  \begin{minipage}[t]{0.3\textwidth}
    \centering
    \includegraphics[width=\linewidth]{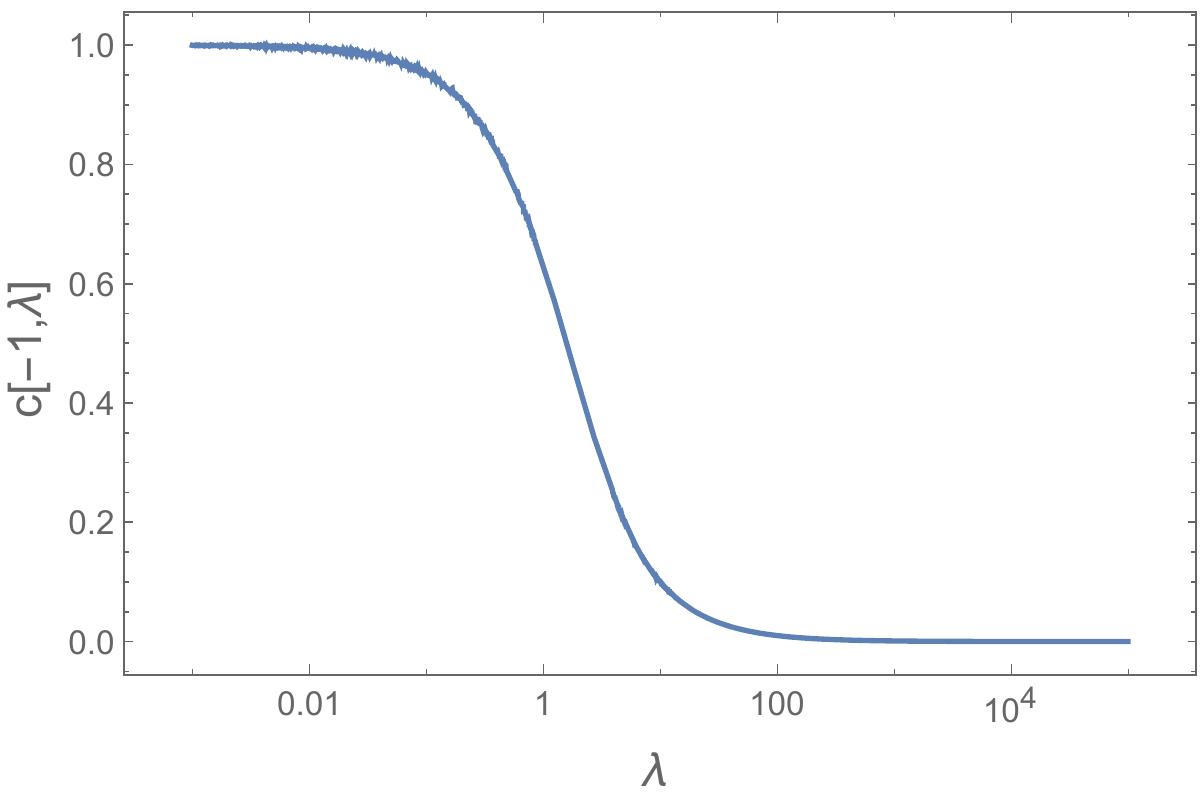}
    \caption*{$\alpha=-1$}
  \end{minipage}%
  \begin{minipage}[t]{0.3\textwidth}
    \centering
    \includegraphics[width=\linewidth]{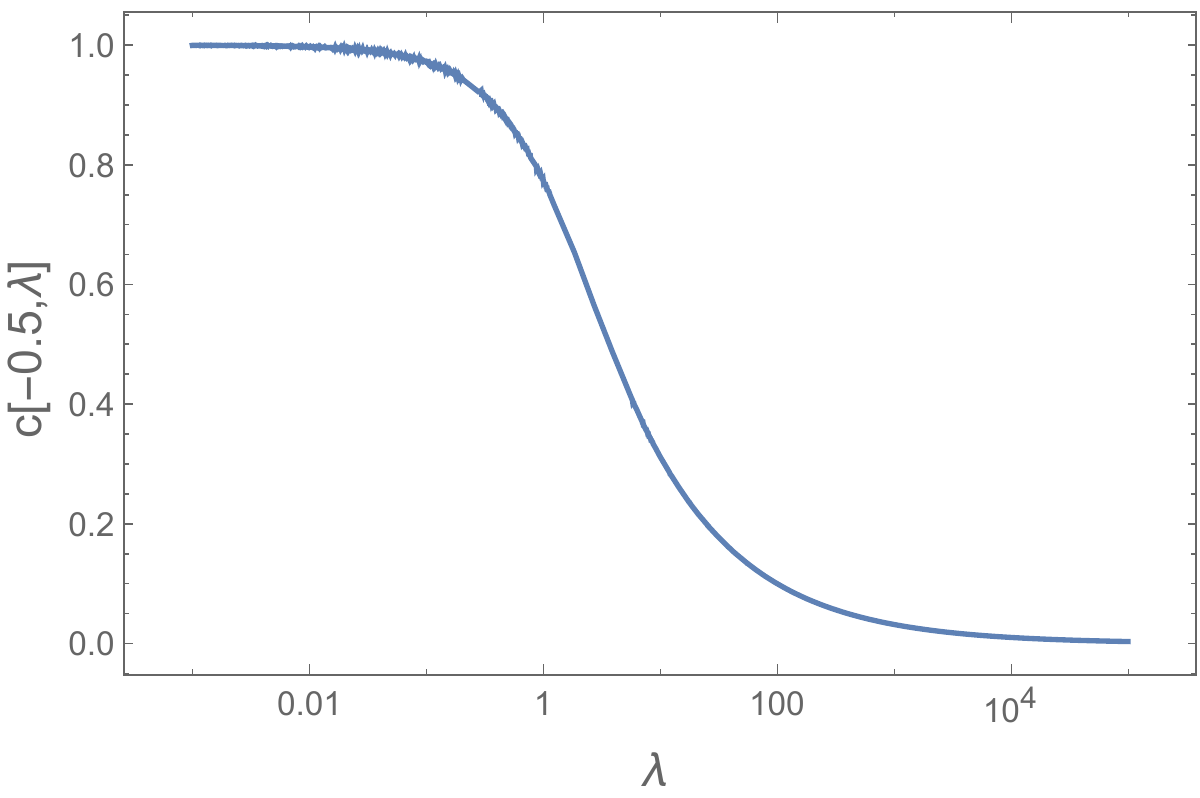}
    \caption*{$\alpha=-0.5$}
  \end{minipage}

  \vspace{0.5cm} 

  \begin{minipage}[t]{0.3\textwidth}
    \centering
    \includegraphics[width=\linewidth]{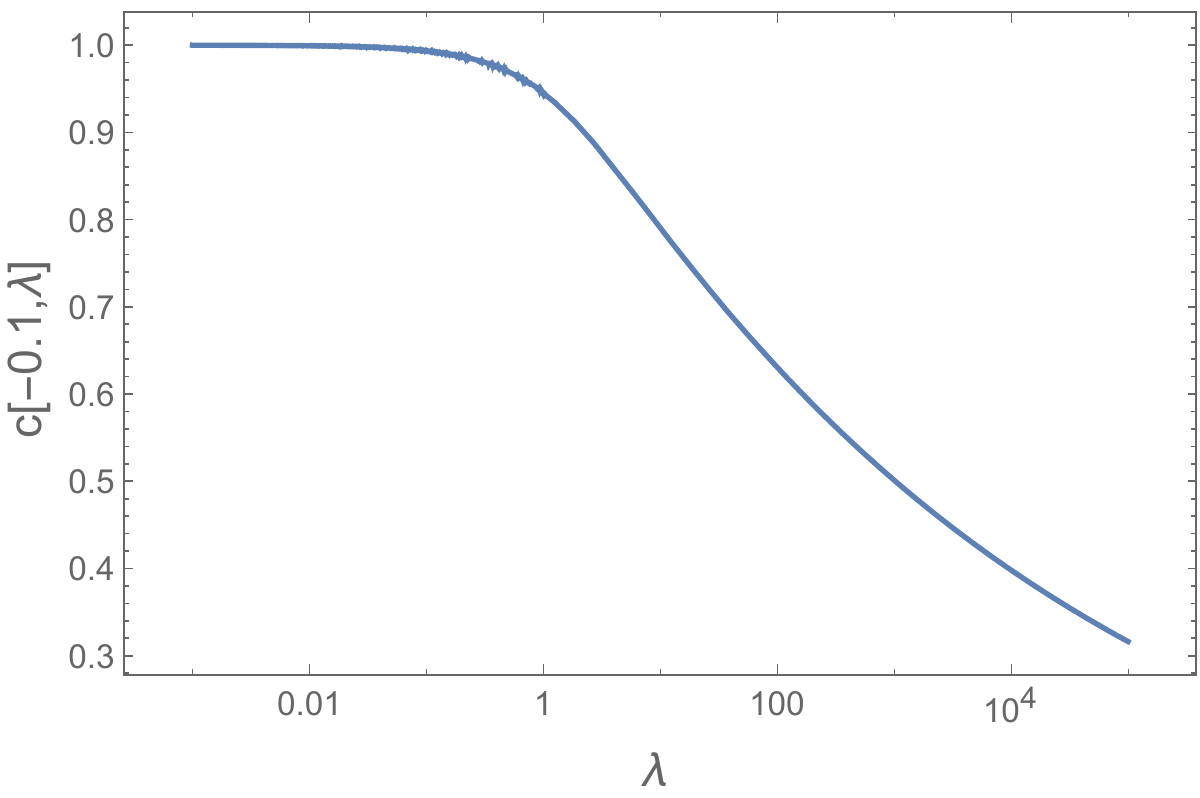}
    \caption*{$\alpha=-0.1$}
  \end{minipage}%
  \begin{minipage}[t]{0.3\textwidth}
    \centering
    \includegraphics[width=\linewidth]{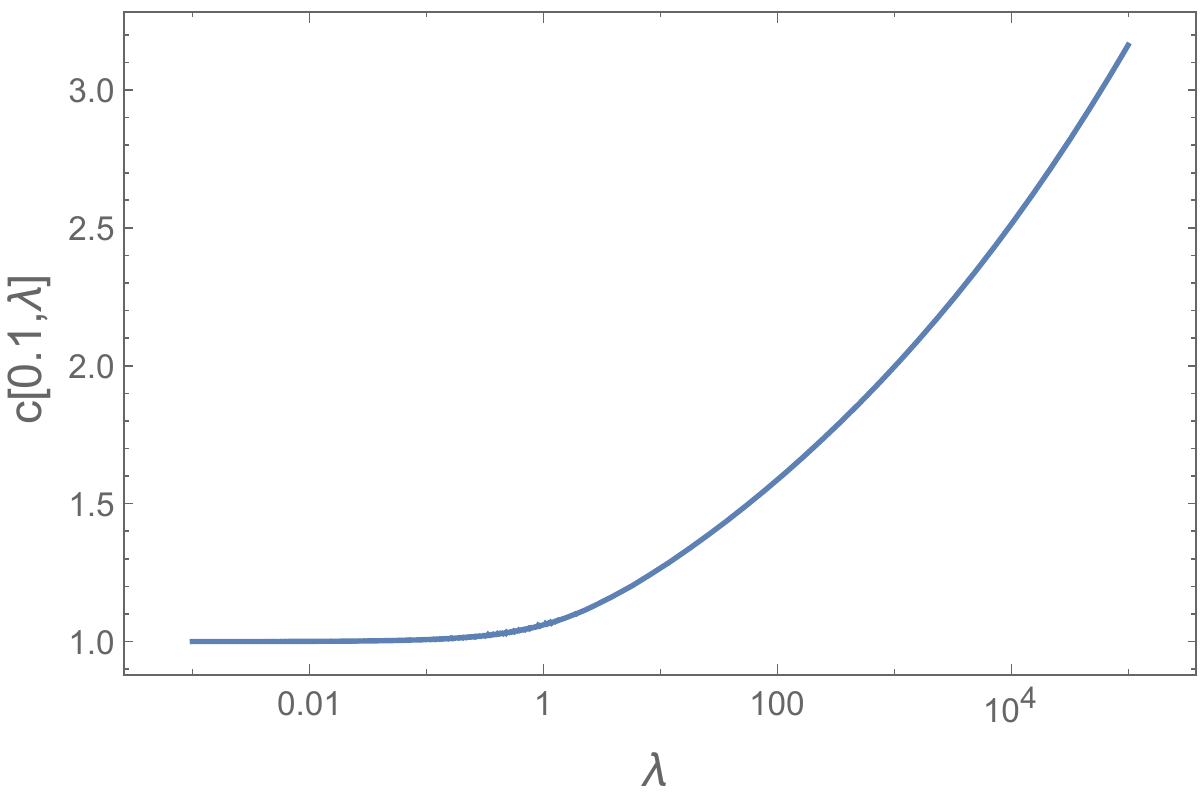}
    \caption*{$\alpha=0.1$}
  \end{minipage}%
  \begin{minipage}[t]{0.3\textwidth}
    \centering
    \includegraphics[width=\linewidth]{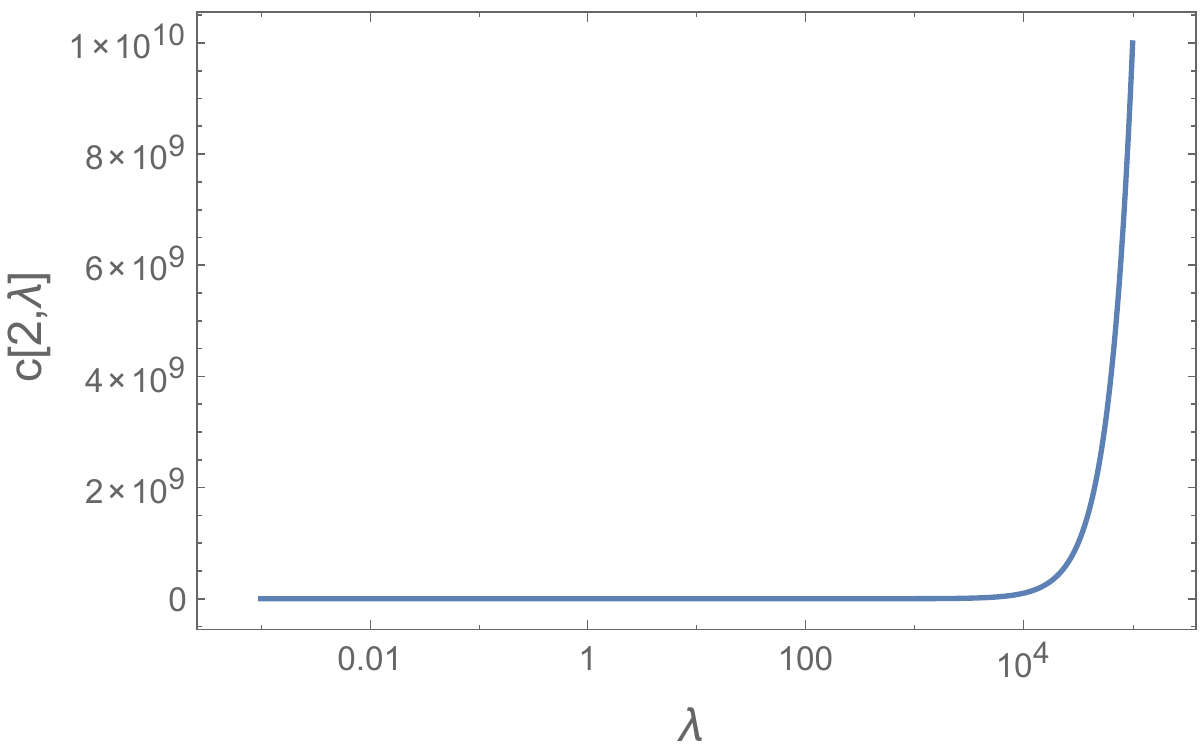}
    \caption*{$\alpha=2$}
  \end{minipage}%

   \caption{ $c[\alpha,\lambda]$ as a function of $\lambda$, where the expectation was computed by averaging over 2000 random variables $Poi(\lambda)$.}
    \label{f2.2}
  
\end{figure}

On the other hand, in a more recent work \cite{Y}, the author employed primarily Taylor expansion and the Mean Value Theorem to study  the generalized Randić index on an  inhomogeneous Erd\H{o}s–Rényi random graph (encompassing the classical Erd\H{o}s–Rényi random graphs). Specifically, in Corollary 2.2, the author obtained that: with $n p_n log(2) \geq log(n)$ (indicating an a.s connected graph), it follows that
$$R_\alpha(G_{n,p})=\frac{n^{2(1+\alpha)} p_n^{1+2 \alpha}}{2}\left[1+O_P\left(\frac{\left(\log \left(n p_n\right)\right)^{4(1-\alpha)_{+}}}{\sqrt{n p_n}}\right)\right].$$
 
In the context provided, $X_n = O_P(a_n)$ signifies that $X_n/a_n$ is bounded in probability.  As expected, adopting  an alternative and  more reduced approach, we obtain more accurate results which  are consistent with these findings. 

Finally, at this point, it is important to emphasize that the identical process (or a slightly different tehcnik) can be utilized for analyze other topological indices derived from degree: $$TI(G)= \displaystyle\sum_{v\sim u} F(d_{v},d_{u}),$$ such as:
\begin{itemize}
    
    \item The atom - bond connectivity ($ABC$), where $F(x, y)=\sqrt{\frac{x+y-2}{x y}}.$
    \item  The geometric - arithmetic index ($GA$), where $F(x, y)=\frac{2\sqrt{xy}}{x+y}.$
     \item The  harmonic index ($H$), where $F(x, y)=\frac{2}{x+y}.$
    \item  The sum - connectivity index ($\chi$), where $F(x, y)=\frac{1}{\sqrt{x+y}}.$
    \item The  logarithm of first multiplicative Zagreb index, where $F(x, y)=2\left(\frac{\ln x}{x}+\frac{\ln y}{y}\right).$

\end{itemize}

As motivation, similarly, the aforementioned approach could offer us either theoretical support or  a contrast for the findings  obtained in \cite{MMRS,DMRSV,DHHIR,LSG,Y}. Another noteworthy aspect is that the main findings may be beneficial to explore topological indices not conforming to the above form. Especially, the energy index of $G$, denoted as $\mathcal{E}(G)$ and defined by

\begin{equation}\label{ed1}
     \mathcal{E}(G)=Tr(|A(G)|),
\end{equation}
can be examined. This is due to the results presented in \cite{AA} and \cite{YLL}, which establish that

$$2R_{-1/2}(G) \leq \mathcal{E}(G) \leq 2\sqrt{\Delta(G)}R_{-1/2}(G) .$$

From that perspective, the implications of the observed long-run behavior of the classical Randić index of the Erd\H{o}s-Rényi sparse random graph suggest the conjecture that $\mathcal{E}(G)/n \to g(\lambda)$ when $n \to \infty$.  By the way, something similar was conjectured for the trees generated by the Barabasi-Albert model with parameter $\alpha$ in \cite{AD}; perhaps, the approach developed in this paper can be explored in the Barabasi-Albert model and coupled with equation (\ref{ed1}) to say something about the asymptotic behavior of the energy. As an additional point, for the dense case, it is well known that almost all ER graphs over $n$ vertices have energy of order $n^{3/2}(\frac{8}{3\pi} \sqrt{p(1-p)}+o(1))$ \cite{DLL}. Actually, since the adjacency matrix of a random graph is a random matrix, the proof relies on
one of the most important achievements in this field which is Wigner's semicircular law. In general, examining the classical Randić index of the Erdős-Rényi random graph can provide insight into the asymptotic trends of energy, such as hyperenergetic and hypoenergetic characteristics, and conversely \citep{EKYY}.

\section*{Funding Information}

Laura Eslava was supported by DGAPA-PAPIIT-UNAM grant IN-102822, Arno Siri-Jégousse was supported by DGAPA-PAPIIT-UNAM grant IN-102824 and Salyé Sigarreta was supported by CONAHCYT grant.

\end{document}